\documentclass{aamas2017}

\usepackage{times}
\usepackage{amsmath}
\usepackage{url}
\usepackage{dsfont}
\usepackage{tabularx}
\usepackage{multirow}
\usepackage{calc}
\usepackage[table]{xcolor}
\usepackage{wrapfig}
\usepackage{tikz}
\usepackage{booktabs}
\usepackage{hyperref}
\usepackage{xspace}
\usepackage{textcase}
\usepackage{graphicx}
\usepackage{color}
\usepackage{pifont}
\usepackage{booktabs}
\usetikzlibrary{positioning}
\usepackage{lineno}
\usepackage{hyphenat}

\DeclareMathAlphabet{\mathpzc}{OT1}{pzc}{m}{it}

\newcommand{\mathparedefine}{{k}}

\newcommand{\mathmainpartdefine}{{I}}
\newcommand{\mathxxxcssize}{\kappa}

\newcommand{\arc}[2]{(#1,#2)}

\newcommand{\fpt}{{{FPT}}}
\newcommand{\w}{W}

\newcommand{\wa}{{\w}[1]}

\newcommand{\wah}{{\w}[1]-hard}

\newcommand{\xp}{{{XP}}}
\newcommand{\wb}{{\w}[2]}
\newcommand{\wbh}{{\w}[2]-hard}
\newcommand{\wbhns}{{\w}[2]-hardness}

\newcommand{\wi}{{\w}[$i$]}
\newcommand{\wih}{{\w}[$i$]-hard}
\newcommand{\poly}{{{P}}}
\newcommand{\np}{NP}

\newcommand{\nph}{{{\np}-hard}}
\newcommand{\nphns}{{\np}-hardness}

\newcommand{\yes}{\sc{Yes}}
\newcommand{\no}{\sc{No}}
\newcommand{\yesins}{{\yes}-instance}
\newcommand{\noins}{{\no}-instance}

\newcommand{\trelation}{\succ}
\newcommand{\myresultemp}[1]{{\textbf{#1}}}
\newcommand{\ts}{\pi} 
\newcommand{\vc}{\varphi}

\newcommand{\tourproperty}[1]{TS-#1}
\newcommand{\copname}{Copeland set}
\newcommand{\topcname}{top cycle}
\newcommand{\Topcname}{Top cycle}
\newcommand{\ucname}{uncovered set}
\newcommand{\score}[3]{score(#1,#2,#3)} 
\newcommand{\eunion}[2]{{#1}+{#2}} 
\newcommand{\cop}[1]{CO(#1)}
\newcommand{\topc}[1]{TC(#1)}
\newcommand{\uc}[1]{UC(#1)}

\newcommand{\comments}[1]{}

\newcommand{\myvspace}[1]{}
\newcommand{\she}{she}

\newcommand{\herself}{herself}

\newcommand{\his}{his}
\newcommand{\myfig}[1]{Figure~\ref{#1}}

\newcommand{\onlyfull}[1]{}
\newcommand{\onlyaaai}[1]{#1}

\newtheorem{theorem}{Theorem}

\newtheorem{lemma}{Lemma}

\begin{document}

\title{Approval Voting with Intransitive Preferences}

\numberofauthors{1}

\author{
\alignauthor
Yongjie Yang\\
\affaddr{Chair of Economic Theory}\\
\affaddr{Saarland University, Saarbr\"{u}cken, Germany}\\
\email{yyongjiecs@gmail.com}
}

\maketitle

\begin{abstract}
We extend Approval voting to the settings where voters may have intransitive preferences. The major obstacle to applying Approval voting in these settings is that voters are not able to clearly determine who they should approve or disapprove, due to the intransitivity of their preferences. An approach to address this issue is to apply tournament solutions to help voters make the decision. We study a class of voting systems where first each voter casts a vote defined as a tournament, then a well-defined tournament solution is applied to select the candidates who are assumed to be approved by the voter. Winners are the ones receiving the most approvals. We study axiomatic properties of this class of voting systems and complexity of control and bribery problems for these voting systems.
\end{abstract}

\keywords{approval; tournament solution; voting system; complexity} 

\section{Introduction}
\label{sec:introduction}
Voting is a common method for preference aggregation and collective decision-making, and has significant applications in
multi-agent systems, political elections, web spam reduction, pattern recognition, etc.~\cite{DBLP:conf/www/DworkKNS01,Egan2014,Kalech2011,DBLP:journals/prl/LuminiN06}. For instance, in multiagent systems, it is often necessary for a group of agents to make a collective decision by means of voting in order to reach a joint goal.
Approval-based voting systems are among the most important voting systems and have been extensively studied in the literature~\cite{AAAI2015AzizBCEFW,baumeisterapproval09,Fishburm81,Kilgour2014Marshall,DBLP:conf/icaart/Lin11,DBLP:conf/aaai/SkowronF15,Yangaamas14b,AAMAS15YangSG}. In an approval-based voting, each voter has a preference over the candidates, and based on the preference, the voter determines a subset of candidates that {\she} approves. The winners are the candidates that get the most approvals. In most of the approval-based voting systems, voters are assumed to have transitive preferences. That is, if a voter prefers a candidate $a$ to another candidate $b$, and prefers $b$ to a third candidate $c$, then the voter prefers $a$ to $c$. Among the most well-studied approval-based voting systems are Approval, $r$-Approval and Plurality. In {\it{Approval voting}}, each voter has a dichotomous preference which can be represented by a partition $(C_1,C_2)$ of the candidates, meaning that the voter prefers every candidate in $C_1$ to every candidate in $C_2$, and are indifferent between candidates in each $C_i$ where $i=1,2$. Moreover, a voter with a dichotomous preference $(C_1,C_2)$ approves all candidates in $C_1$ and disapproves  all candidates in $C_2$. In {\it{$r$-Approval voting}}, each voter has a preference which is defined as a linear order over the candidates, and approves exactly the top-$r$ ordered candidates. {\it{Plurality}} is exactly $1$-Approval.

\myvspace{-5pt}
\subsection{Motivation}
There is no doubt that transitive preferences occur naturally in many real-world applications. The question is whether intransitive preferences make sense either. The answer is ``Yes!''. In fact, there exist many real-world applications where voters may have intransitive preferences (see, e.g.,~\cite{DBLP:conf/uai/Elkind014,DBLP:journals/jair/FaliszewskiHHR09,DBLP:conf/atal/FaliszewskiHS10,Monjardet1978,SaarinenTG2014AAAIworkshopIntransitivePreference}). For instance, when voters compare candidates based on, not one, but multiple quality parameters~\cite{DBLP:conf/uai/Elkind014,DBLP:conf/atal/XiaCL10}. Another natural scenario where intransitive preferences arise is that when the number of candidates is considerably large~\cite{DBLP:conf/uai/Elkind014}. In this case, it is more efficient to utilize vote elicitation techniques where voters iteratively cast parts of their preference such as pairwise comparisons. In addition, intransitive preferences may also arise in the settings of sport prediction, where an agency (e.g., a gambling company) desires to predict the sport result for some special purpose. In order to gain a result as precise as possible, the agency might resort to several experts, who based on their expertise suggest the winning player in each pending match between two players. Then, based on the suggestions, the agency applies a voting to predict the winning player(s). In this case, each suggestion by an expert may be considered as a preference which is not necessarily transitive, since it is commonplace to see that a player $a$ who beats another player $b$ is beaten by a third player $c$ who is beaten by $b$. We refer to~\cite{KirchsteigerP1996Intransitive,Tversky1969intranstivepreferences} for further discussion with several concrete examples.

Finally, we would like to point out that intransitive preferences may also occur in district-based voting or group-based voting (a group may be a political party, a department in a university, an affiliate of a company, a set of robots, etc.), where each group consists of their own group members and is only allowed to submit one single group vote. In such a situation, group leaders may need to first apply a voting to derive their group vote before the whole voting. If there are three group members whose preferences over three candidates $a,b,c$ are $a\succ b\succ c, b\succ c\succ a, c\succ a\succ b$, respectively, then the group vote would probably be $a\succ b, b\succ c$ but $c\succ a$, an intransitive preference. A significant difference between this case and the cases mentioned above is that each intransitive preference (group vote) in this case is drawn from the votes cast by the group members.   However, each previous mentioned intransitive preference is cast by a single voter.

We extend the framework of approval-based voting to the settings where voters may hold intransitive preferences over the candidates. The major difficulty of imposing the framework in these settings is that voters with intransitive preferences are not able to determine who they should approve. In order to address this issue, we utilize tournament solutions. It should be noted that an intransitive preference can be represented by a tournament (we consider only complete preferences, i.e., for every two candidates $a$ and $b$ either $a$ is preferred to $b$ or the other way around)---a complete and asymmetric binary relation. From the graph theory point of view, a tournament is a {\it{directed graph}} where there is exactly one arc between every pair of vertices (candidates). A {\it{tournament solution}} is a function that maps each tournament to a nonempty subset of candidates.
Tournament solutions as a powerful decision making model have been extensively-studied in the literature~\cite{handbookofcomsoc2015Cha3Brandt,Brandt2010,BrandtS2014ORDiscrimitiveTournament,MSYangIJCAI2015,Yang2016FurtherstepTES}. For instance, tournament solutions have significant applications in voting. In particular, given a set of votes, we can create a tournament based on the majority relations between the candidates (assume that the number of votes is odd): create an arc from $a$ to $b$ if there are more voters preferring $a$ to $b$. Then, a tournament solution is applied to the tournament to select the winners. 

Each approval-based voting system studied in this paper is a natural combination of the classic Approval voting and a well-studied tournament solution. Precisely, fixing a tournament solution, each voter in this setting submits a vote which is defined as a tournament. Then, every candidate selected by the tournament solution is approved by the voter, and every candidate not selected is disapproved. We remark that in practice, voters need only to cast their votes, but leave the duty of calculating the winning candidates with respect to the tournament solution to a second agency (e.g., a computer, the voting designer, etc.), since it is unnatural to assume or require that every voter knows how the tournament solution works. In other words, we assume that the voters implicitly approve the winning candidates in their cast tournaments, with respect to the associated tournament solution.
In this paper, we mainly consider three tournament solutions, namely, the top cycle, Copeland set and uncovered set. One reason that we choose them to study is that they are among the most prevalent tournament solutions which have been extensively studied in the literature. 
It also makes sense to consider other tournament solutions such as minimal covering set and tournament equilibrium set (see~\cite{handbookofcomsoc2015Cha3Brandt} for further tournament solutions).

As we pointed out earlier in the example on group-based voting, tournaments in our model are not necessarily to be cast directly by voters, but can be also drawn from the majority relations between the candidates according to the votes cast by some group members. In this scenario, it makes much sense to first apply a tournament solution to the group votes (tournaments) to determine the winning candidates in each subvoting.

We would like to point out that apart from the model we proposed in this paper, there are several other prominent approaches to aggregate tournaments. For instance, we could apply different tournament solutions to the given tournaments to determine the candidates implicitly approved by the voters. An explanation is that in a group-based voting, each group (leader) is allowed to freely choose a tournament solution to use. Another approach to select winners from a set of tournaments would be as follows: First, we create a tournament based on the majority relations between the candidates, i.e., there is an arc from $a$ to $b$ if there is an arc from $a$ to $b$ in a majority of the given tournaments. Then, we apply a tournament solution to the tournament to select the winners. In addition, researchers have studied the model of deriving a ranking of the candidates based on a give set of tournaments, see, e.g.,~\cite{Monjardet1978} and references therein.

\myvspace{-2pt}
\subsection{Our Contribution}
The major contribution of this work is the initialization of the study of a class of voting systems for the scenarios where voters may have intransitive preferences. To give a comprehensive understanding of the new voting systems, we study some axiomatic properties of these voting systems. Axiomatic properties of voting systems are primary factors used to evaluate voting systems and important guidance for voting organizers to select a proper voting system for their specific purpose in practice~\cite{DBLP:conf/atal/ElkindFSS14,DBLP:conf/aaai/FreemanBC14,DBLP:journals/scw/Schulze11,DBLP:conf/sigecom/Xia15}. In particular, we prove that these voting systems satisfy several axiomatic properties for many common tournament solutions. As a byproduct, we introduce two concepts of monotonicity for tournament solutions, and show that the top cycle satisfies both monotonicity criteria, while both the Copeland set and the uncovered set averse to the monotonicity criteria. Our results concerning axiomatic properties are summarized in Theorems~\ref{thm_all_consistent}-\ref{thm_TC_Approval_Pareto_CO_UC_Approval_not}.

In addition, we study the complexity of strategic behavior under these voting systems. In particular, we study control and bribery problems. We achieve polynomial-time solvability results, {\nphns} results as well as {\wbhns} results. See Table~\ref{tab:complexitysummary} for a summary of these results. A general conclusion is that these voting systems resist more types of strategic behavior than other approval-based voting systems such as Plurality and Approval. Studying complexity of strategic voting problems has been the main focus of many research papers. First, complexity is widely considered as a prominent theoretical barrier against strategic behavior in voting systems. Second, complexity helps practitioners decide what
kind of solution method is appropriate. For polynomial-time solvability results, we directly provide efficient algorithms with low time complexity. On the other hand, hardness results (e.g.,~{\nphns} results and {\wbhns} results) suggest that finding an exact solution is apt to be costly or impractical, and resorting to approximation  algorithms or heuristic algorithms may be a necessary choice.
Finally, it should be pointed out that complexity of strategic behavior for voting systems has also been considered as an important factor to evaluate voting systems, see, e.g., ~\cite{Bartholdi92howhard}.

\myvspace{-8pt}
\section{Preliminaries}\label{sec:preliminaries}

{\textbf{Tournament.}} 
In this paper, we use the terms ``candidate'' and ``vertex'' interchangeably.
A \emph{tournament} $T$ is a pair $(V(T), \trelation)$ where $V(T)$ is a set of {\it{candidates}} and $\trelation$ is an asymmetric and complete binary relation on $V(T)$.
For $X,Y\subseteq V(T)$ such that $X\cap Y=\emptyset$, $X \trelation Y$ means that $x\trelation y$ for every $x\in X$ and every $y\in Y$.
For ease of exposition, we use {\it{directed graphs}} to represent tournaments. Precisely, in this paper a tournament $T=(V(T), \trelation)$ is considered as a directed graph where $V(T)$ is considered as the {\it{vertex set}} and $\trelation$ is considered as the {\it{arc set}}. We refer to the textbook by West~\cite{Douglas2000} for readers who are not familiar with graph theory.

For a candidate $a\in V(T)$, let $N^-_{T}(a)$ denote the set of {\it{inneighbors}} of $a$ in $T$ and $N^{+}_{T}(a)$ the set of {\it{outneighbors}} of $a$, i.e., $N^-_{T}(a)=\{b\in V(T)\mid b\trelation a\}$ and $N^+_{T}(a)=\{b\in V(T)\mid a\trelation b\}$. The {\it{indegree}} (resp. {\it{outdegree}}) of $a$ is defined as $|N^-_T(a)|$ (resp. $|N^+_T(a)|$). A tournament $T$ is {\it{regular}} if for every candidate $a$ in $T$ it holds that $||N_T^+(a)|-|N_T^-(a)||\leq 1$. A {\it{directed triangle}} is a regular tournament with three vertices.

A {\it{directed path}} from a candidate $a$ to another candidate $b$ is a vertex sequence $(a=v_1,v_2,...,v_t=b)$ such that $v_i\trelation v_{i+1}$, for every $i\in \{1,2,...,t-1\}$. A tournament $T$ is {\it{strongly connected}} if there is a directed path from every candidate to every other candidate.
A {\it{maximal strongly connected component}} of $T$ is a strongly connected subtournament of $T$ with maximal vertices.

The {\it{source}} of a tournament $T=(V(T),\trelation)$ is the candidate $a$ so that $a\trelation b$ for every candidate $b\in V(T)\setminus \{a\}$.
The source is also called the {\it{Condorcet winner}} of the tournament from the social choice point of view. Clearly, not every tournament has a source. 
For a subset $B\subseteq V(T)$, $T[B]$ is the {\it{subtournament}} induced by $B$, i.e., $T[B]=(B,\trelation')$ where for every $a,b\in B$, $a\trelation' b$ if and only if $a\trelation b$.

{\textbf{Tournament Solution.}} A {\it{tournament solution}} $\ts$ is a function that maps every tournament $T$ to a nonempty subset $\ts(T)\subseteq V(T)$. In this paper, we mainly study the following three tournament solutions~\cite{handbookofcomsoc2015Cha3Brandt}. We refer to~\cite{handbookofcomsoc2015Cha3Brandt} for a comprehensive introduction to further well-studied tournament solutions.

\myvspace{-6pt}
\begin{description}\itemsep=-2pt
\item[\it{Copeland Set}.] The {\it{Copeland score}} of a candidate $c$ in a tournament $T$
is defined as the outdegree of $c$ in $T$. 
The Copeland set of $T$, denoted by $\cop{T}$, consists of all candidates with the highest Copeland score.

\item[\it{Top Cycle}.]  The top cycle $\topc{T}$ of a tournament $T$ is the unique minimal subset of candidates such that there is an arc from every candidate in the subset to every candidate not in the subset.

\item[\it{Uncovered Set}.] A {\it{king}} in a tournament is a candidate $a$ such that for every other candidate $b$, either $a\trelation b$ or there is another candidate $c$ such that $a\trelation c$ and $c\trelation b$. It is folklore that every tournament has at least one king~\cite{Landau1953tournamentking}. The uncovered set of a tournament $T$, denoted by $\uc{T}$, is the set of all kings of $T$.
\end{description}

It is known that for every tournament $T$ it holds $\cop{T},\uc{T}\subseteq \topc{T}$~(see, e.g.,~\cite{handbookofcomsoc2015Cha3Brandt}).

{\bf{Election.}} An {\it{election}} is a tuple $\mathcal{E}=(\mathcal{C},\mathcal{T})$,
where $\mathcal{C}$ is a set of candidates, and $\mathcal{T}$ is a list of {\it{votes}}
(for convenience, the terminologies ``vote'' and ``voter'' are used interchangeably throughout this paper).
In this paper, we consider only the election where each vote is defined as a tournament $T=(\mathcal{C},\trelation)$. For two candidates $a,b\in \mathcal{C}$
and a vote $T(\mathcal{C},\trelation)$, $a\trelation b$ means that the vote prefers $a$ to $b$.
A {\it{voting correspondence}} 
$\vc$ is a function that maps an election $\mathcal{E}=(\mathcal{C},\mathcal{T})$ to a
nonempty subset $\vc(\mathcal{E})$ of $\mathcal{C}$. We call the elements in $\vc(\mathcal{E})$
the {\it{winners}} of $\mathcal{E}$ with respect to $\vc$. If $\vc(\mathcal{E})$ consists of only one winner, we call
it the {\it{unique winner}}; otherwise, we call them {\it{co-winners}}. For two non-overlapping lists of tournaments
$\mathcal{T}=(T_1,T_2,...,T_x)$ and
$\mathcal{T}'=(T_1',T_2',...,T_y')$, we denote by $\mathcal{T}+\mathcal{T}'$ the list $(T_1,...,T_x,T_1',...,T_y')$. For two elections
$\mathcal{E=(C,T)}$ and $\mathcal{E'=(C,T')}$ with the same candidate set $\mathcal{C}$,
 $\eunion{\mathcal{E}}{\mathcal{E}'}=(\mathcal{C},\mathcal{T}+\mathcal{T}')$.

{\bf{Implicit Approval Voting}}.
Now we introduce the core concept in this paper---$\ts$-Approval. Each $\ts$-Approval is a combination of the prevalent Approval voting and a tournament solution $\ts$. To the best of our knowledge, such voting correspondences have not been studied in the literature.
Let $\mathcal{E}=(\mathcal{C},\mathcal{T})$ be an election.

\myvspace{-8pt}
\begin{center}
\begin{tabular}{|p{0.45\textwidth}|}\hline
{\textbf{$\ts$-Approval}} 

Each candidate $c\in \mathcal{C}$ is assigned a score defined as $\score{c}{\mathcal{E}}{\ts}=|\{T\in \mathcal{T}\mid c\in \ts(T)\}|$.
The candidates with the highest score are the winners.\\ \hline
\end{tabular}
\end{center}
\myvspace{-2pt}

In Approval, each voter explicitly determines {\herself} whom {\she} wants to approve.
In $\ts$-Approval, each voter with preference $T$ is assumed to implicitly approve
all candidates in $\ts(T)$ and disapprove all the remaining candidates.


{\bf{Properties of Voting Correspondences.}}

\myvspace{-5pt}
\begin{description}\itemsep=-2pt
\item[\it{Anonymity}.] A voting correspondence $\vc$ is anonymous if reordering the votes does not affect the winning set. That is, for every two elections $\mathcal{E}=(\mathcal{C},\mathcal{T}=(T_1,...,T_n))$ and $\mathcal{E}'=(\mathcal{C},\mathcal{T}'=(T_{\sigma(1)}, T_{\sigma(2)},...,T_{\sigma(n)}))$ where $(\sigma(1), \sigma(2),...,\sigma(n))$ is a permutation of $(1,2,...,n)$, it holds that $\vc(\mathcal{E})=\vc(\mathcal{E}')$.

\item[\it{Neutrality}.] An election $(\mathcal{C},\mathcal{T}=(T_1,...,T_n))$ is {\it{isomorphic}} to
another election $(\mathcal{C}',\mathcal{T}'=(T_1',...,T_n'))$ where $T_i=(\mathcal{C},\trelation_i)$ and $T_i'=(\mathcal{C}',\trelation_i')$ for every $i\in \{1,...,n\}$,
 if there is an one-to-one mapping
$f: \mathcal{C}\mapsto \mathcal{C}'$ such that
for every two distinct candidates $a,b\in \mathcal{C}$ and every $i\in \{1,2,....,n\}$, it holds that $a\trelation_i b$ if and only if $f(a)\trelation_{i}' f(b)$.
A voting correspondence $\vc$ is neutral if for every two isomorphic elections
$\mathcal{E}=(\mathcal{C},\mathcal{T})$ and $\mathcal{E}'=(\mathcal{C}',\mathcal{T}')$, and every $c\in \mathcal{C}$,
it holds that $c\in \vc(\mathcal{E})$ if and only if $f(c)\in \vc(\mathcal{E}')$, where $f$ is the mapping as discussed above for
$\mathcal{E}$ and $\mathcal{E}'$.

\item[\it{Monotonicity}.] A voting correspondence $\vc$ is monotonic if for every two elections
$\mathcal{E}=(\mathcal{C},\mathcal{T}=(T_1,...,T_n))$, $\mathcal{E}'=(\mathcal{C},\mathcal{T}'=(T_1',...,T_n'))$,
and every $c\in \vc(\mathcal{E})$ such that for every $i\in \{1,2,...,n\}$
    (1) $T_i[\mathcal{C}\setminus \{c\}]=T_i'[\mathcal{C}\setminus \{c\}]$; and
    (2) $N^+_{T_i}(c)\subseteq N^+_{T_i'}(c)$,
    it holds that $c\in \vc(\mathcal{E}')$.


\item[\it{Majority}.] A voting correspondence $\vc$ satisfies the majority criterion if for every
election $\mathcal{E}=(\mathcal{C},\mathcal{T})$ where there is a candidate $c\in \mathcal{C}$ which is the source
in a majority of the tournaments in $\mathcal{T}$, it holds that $c\in \vc(\mathcal{E})$.

\item[\it{Consistency}.] A voting correspondence $\vc$ is consistent if for every two elections
$\mathcal{E=(C,T)}$ and $\mathcal{E'=(C,T')}$, it holds that
$\vc(\mathcal{E})\cap \vc(\mathcal{E}')\subseteq \vc(\eunion{\mathcal{E}}{\mathcal{E}'})$.

\item[\it{Pareto optimal}.] A voting correspondence ${\vc}$ is Pareto optimal if for every election $\mathcal{E=(C,T)}$ and every two candidates $a,b\in \mathcal{C}$ such that $a\succ b$ in every tournament $T=(\mathcal{C},\succ)\in \mathcal{T}$, $a\not\in \vc(\mathcal{E})$ implies $b\not\in \vc(\mathcal{E})$.
\end{description}

{\bf{Properties of Tournament Solutions.}}
To the names of the following properties of tournament solutions, we append a prefix ``TS'', standing for ``tournament solution'', to avoid confusion with the definitions of the axiomatic properties of voting correspondences.

\myvspace{-5pt}
\begin{description}\itemsep=-2pt
\item[\it{\tourproperty{Neutrality}}.] Two tournaments $T=(\mathcal{C},\trelation)$
and $T'=(\mathcal{C}',\trelation')$ where $|\mathcal{C}|=|\mathcal{C}'|$ are {\it{isomorphic}} if there is an one-to-one
mapping $f: \mathcal{C}\mapsto \mathcal{C}'$ such that for every two $a,b\in \mathcal{C}$,
it holds that $a\trelation b$ if and only if $f(a)\trelation' f(b)$. Here, $f$ is called an {\it{isomorphic mapping}} of $T$ and $T'$.
A tournament solution $\ts$ satisfies the {\tourproperty{neutrality}} criterion if for every two isomorphic tournaments
$T=(\mathcal{C},\trelation)$ and $T'=(\mathcal{C}',\trelation')$, it holds that
$\ts(T')=\{f(a)\in \mathcal{C}'\mid a\in \ts(T)\}$, where $f$ is an isomorphic mapping of $T$ and $T'$.

\item[\it{\tourproperty{Monotonicity}}.] A tournament solution $\ts$ is {\tourproperty{monotonic}} if for every two tournaments $T=(\mathcal{C},\succ)$, $T'=(\mathcal{C},\succ')$,
and every candidate $c\in \ts(T)$ such that
$T[\mathcal{C}\setminus \{c\}]=T'[\mathcal{C}\setminus \{c\}]$ and $N^+_T(c)\subseteq N^+_{T'}(c)$,
it holds that $c\in \ts(T')$.

\item[\it{\tourproperty{Condorcet consistency}}.] A tournament solution ${\ts}$
is {\tourproperty{Condorcet consistent}} if for every
tournament $T$ which admits the Condorcet winner $w$, $\ts(T)=\{w\}$.
\end{description}

Now we introduce two concepts of monotonicity of tournament solutions. To the best of our knowledge, they have not been studied in the literature. Generally speaking, a tournament solution is {\it{{\tourproperty{exclusive monotonic}}}} if a winning candidate $c$ remains as a winning candidate when
$c$ is preferred to more candidates, without changing the preferences between other candidates. Moreover,
no nonwinning candidate benefits from this, i.e., no nonwinning candidate becomes a winning candidate. So, if a tournament solution is {\tourproperty{exclusive monotonic}, then making a winning candidate stronger never makes a nonwinning candidate better off. The formal definition is as follows.

\begin{description}
\item[\it{\tourproperty{Exclusive monotonicity}}.] A tournament solution $\ts$ is
{\tourproperty{exclusive monotonic}} if for every two tournaments
$T=(\mathcal{C},\succ)$, $T'=(\mathcal{C},\succ')$, and every  $c\in \ts(T)$ such that
    (1) $T[\mathcal{C}\setminus \{c\}]=T'[\mathcal{C}\setminus \{c\}]$; and
    (2) $N^+_T(c)\subseteq N^+_{T'}(c)$,
    it holds that $c\in \ts(T')$ and $\ts(T')\subseteq \ts(T)$.
\end{description}

A tournament solution is {\it{{\tourproperty{exclusive negative monotonic}}}} if when
a nonwinning candidate $c$ is preferred to more candidates, and some other nonwinning candidate becomes a winning candidate, then $c$ must become a winning candidate as well. 
In other words, if extending the outneighborhood of $c$ benefits some nonwinning candidates, then $c$ must benefit from this operation. The formal definition is as follows.

\begin{description}
\item[\it{\tourproperty{Exclusive negative monotonicity}}] (\tourproperty{ENM}).
A tournament solution $\ts$
satisfies the {\tourproperty{ENM}} criterion,
if for every two tournaments $T=(\mathcal{C},\succ)$ and $T'=(\mathcal{C},\succ')$,
and every candidate $c\not\in \ts(T)$ such that
    (1) $T[\mathcal{C}\setminus \{c\}]=T'[\mathcal{C}\setminus \{c\}]$; and
    (2) $N^+_T(c)\subseteq N^+_{T'}(c)$,
    it holds that $\ts(T')\not\subseteq \ts(T)$ implies $c\in \ts(T')$.
\end{description}

{\bf{Parameterized Complexity}}.
A {\it{parameterized problem}} is a language $L \subseteq \Sigma^*\times\mathbb{N}$, where $\Sigma$ is a finite alphabet.
The first \mbox{component} $\mathmainpartdefine\in \Sigma^*$ is called the {\it{main part}}, and the second component $\mathparedefine\in \mathbb{N}$ is called the {\it{parameter}}. Downey and Fellows \cite{fellows99} established the parameterized complexity theory and developed the following parameterized complexity hierarchy:
\[\text{\fpt} \subseteq \text{\wa}\subseteq \text{\wb}\subseteq...\subseteq \text{\xp}.\]
In particular, {\fpt} (stands for fixed-parameter tractable) includes all parameterized problems which admit $O(f(\mathparedefine)\cdot|\mathmainpartdefine|^{O(1)})$-time algorithms. Here $f(\mathparedefine)$ is a computable function of $k$ and $|I|$ is the size of the main part.
Given two parameterized problems $Q$ and $Q'$, an {{{\fpt}-{\it{reduction}}}} from $Q$ to $Q'$ is an
algorithm that takes as input an instance $(\mathmainpartdefine,\mathparedefine)$ of $Q$ and outputs an instance $(\mathmainpartdefine',\mathparedefine')$ of $Q'$ such that 

(1) the algorithm runs in $f(\mathparedefine)\cdot |\mathmainpartdefine|^{O(1)}$ time, where $f$ is a computable function in $\mathparedefine$;

(2) $(\mathmainpartdefine,\mathparedefine)\in Q$ if and only if $(\mathmainpartdefine',\mathparedefine')\in Q'$; and

(3) $\mathparedefine'\leq g(\mathparedefine)$, where $g$ is a computable function in $\mathparedefine$.

A problem is {\wih} for a positive integer $i$ if all problems in {\wi} can be {\fpt}-reducible to the problem.
{\wih} problems are unlikely to admit {\fpt}-algorithms, unless the parameterized complexity hierarchy
collapses at some level~\cite{fellows99}.

\myvspace{-5pt}
\section{Axiomatic Properties}
In this section, we study axiomatic properties for $\ts$-Approval for different tournament solutions $\ts$.
It is fairly easy to check that $\ts$-Approval is anonymous for all tournament solutions $\ts$.
Moreover, $\ts$-Approval is neutral for all tournament solutions
$\ts$ which satisfy the {\tourproperty{neutrality}} criterion.
Furthermore, $\ts$-Approval satisfies the majority criteria for all $\ts$ that are {\tourproperty{Condorcet consistent}}. We now study some other properties for $\ts$-Approval. Consider first the consistency criterion.

\begin{theorem}
\label{thm_all_consistent}
$\ts$-Approval is consistent for all tournament solutions $\ts$.
\end{theorem}
\begin{proof}
Let $\mathcal{E}_1=(\mathcal{C},\mathcal{T}_1)$ and $\mathcal{E}_2=(\mathcal{C},\mathcal{T}_2)$
be two elections such that $\ts(\mathcal{E}_1)\cap \ts(\mathcal{E}_1)\neq \emptyset$.
For a candidate $c\in \mathcal{C}$, it holds that
\[\score{c}{\mathcal{E}_1+\mathcal{E}_2}{\ts}=\score{c}{\mathcal{E}_1}{\ts}+\score{c}{\mathcal{E}_2}{\ts}.\]
This directly implies that if a candidate $c\in \mathcal{C}$ has the highest score in both $\mathcal{E}_1$ and $\mathcal{E}_2$, then $c$ has the highest score in the combined election $\mathcal{E}_1+\mathcal{E}_2$. It then follows that $\ts(\eunion{\mathcal{E}_1}{\mathcal{E}_2})=\ts(\mathcal{E}_1)\cap \ts(\mathcal{E}_2)$. 
\end{proof}

Now we study the monotonicity of $\ts$-Approval for all {\tourproperty{Condorcet consistent}} tournament solutions $\ts$. It should be noted that almost all commonly used tournament solutions, including all tournament solutions studied in this paper, are {\tourproperty{Condorcet consistent}}. We derive both sufficient and necessary conditions for such $\ts$-Approval to be monotonic.

\myvspace{-4pt}
\begin{theorem}
\label{thm:concorcetconsistentmonotonicity}
Let $\ts$ be a {\tourproperty{Condorcet consistent}} tournament solution.
Then, $\ts$-Approval is monotonic if and only if $\ts$ satisfies the
{\tourproperty{exclusive monotonicity}} and {\tourproperty{ENM}} criteria.
%
\end{theorem}

{\begin{proof}
Let $\ts$ be a {\tourproperty{Condorcet consistent}} tournament solution as stated in the theorem and $\varphi=\ts$-{\text{Approval}}.
Assume that $\ts$ satisfies the {\tourproperty{exclusive monotonicity}}
and the {\tourproperty{ENM}} criteria.
Let  $\mathcal{E}=(\mathcal{C},\mathcal{T}=(T_1,...,T_n))$ and
$\mathcal{E}'=(\mathcal{C},\mathcal{T}'=(T_1',...,T_n'))$ be two
elections with the same candidate set $\mathcal{C}$.
Moreover, let $c\in \vc(\mathcal{E})$ be a candidate such that
(1) $T_i[\mathcal{C}\setminus \{c\}]=T_i'[\mathcal{C}\setminus \{c\}]$; and
(2) $N^+_{T_i}(c)\subseteq N^+_{T_i'}(c)$ for every $i\in \{1,2,...,n\}$.
We shall show that $c\in \vc(\mathcal{E}')$. Let's first study the scores of
the candidates in $\mathcal{E}'$.
Apparently, $\score{c}{\mathcal{E}}{\ts}\geq \score{c'}{\mathcal{E}}{\ts}$ for every $c'\in \mathcal{C}\setminus \{c\}$.
Since $\ts$ is {\tourproperty{exclusive monotonic}}, if $c\in \ts(T_i)$
for some $T_i\in \mathcal{T}$, then $c\in \ts(T_i')$. 
For each $a\in \mathcal{C}$, let $\mathcal{W}_a=\{i\in \{1,2,...,n\}\mid a\not\in \ts(T_i), a\in \ts(T_i')\}$.
The following claim is useful.

Claim. $\mathcal{W}_{c'}\subseteq \mathcal{W}_c$ for every $c'\in \mathcal{C}\setminus \{c\}$.

Let $c'$ be a candidate in $\mathcal{C}\setminus \{c\}$, and $T_i$ and $T_i'$ be two tournaments in $\mathcal{T}$ and $\mathcal{T}'$, respectively, such that $c'\not\in \ts(T_i)$ and $c'\in\ts(T_i')$.
Clearly, $\ts(T')\not\subseteq \ts(T)$.
Since $\ts$ is {\tourproperty{exclusive monotonic}}, it must be that $c\not\in \ts(T_i)$; since otherwise,
$\ts(T')\subseteq \ts(T)$, a contradiction. Then, since $\ts$ satisfies
the {\tourproperty{ENM}} criterion, we know that $c\in \ts(T')$. The claim follows.

Due to the above claim and discussions, for every $c'\in \mathcal{C}\setminus\{c\}$, 
\myvspace{-10pt}

\begin{equation*}
\begin{split}
\score{c}{\mathcal{E}'}{\ts}& =\score{c}{\mathcal{E}}{\ts}+|\mathcal{W}_c| \\
& \geq \score{c'}{\mathcal{E}}{\ts}+|\mathcal{W}_c| \\
& \geq \score{c'}{\mathcal{E}}{\ts}+|\mathcal{W}_{c'}|\\
& \geq \score{c'}{\mathcal{E}'}{\ts}.\\
\end{split}
\end{equation*}
\myvspace{-10pt}

\noindent Thus, $c\in \vc(\mathcal{E}')$. 

It remains to prove the other direction. Assume that $\ts$ is not {\tourproperty{exclusive monotonic}}. Then, there exist two tournaments $T$ and $T'$ over the same candidate set $\mathcal{C}$ and a $c\in \ts(T)$ such that
(1) $T[\mathcal{C}\setminus \{c\}]=T'[\mathcal{C}\setminus \{c\}]$;
(2)  $N^+_T(c)\subseteq N^+_{T'}(c)$; and
(3) $c\not\in \ts(T')$, or $c\in \ts(T')$ but $\ts(T')\not\subseteq \ts(T)$.
If $c\not\in \ts(T')$ in Condition~(3), then, the election consisting of only one
vote defined as $T$ obviously shows that $\ts$-Approval is not monotonic. Otherwise,
we construct an election as follows. Let $b$ be any arbitrary candidate in $\ts(T')\setminus \ts(T)$.
The election consists of the following votes: 1 vote defined as $T$; 2 votes each defined as a tournament where
$b$ is the source; 1 vote defined as a tournament where $c$ is the source. Apparently, both $b$ and $c$ are winners in
the election, with each having two approvals. However, by replacing the vote defined as $T$ by $T'$, $b$ gets one more approval from $T'$,
implying $c$ is no longer a winner. Therefore, in this case $\ts$-Approval is not monotonic.

Assume that $\ts$ does not satisfy the {\tourproperty{ENM}}
criterion. Then, there exist two tournaments $T$ and $T'$ and a $c\not\in \ts(T)$ such that
(1) $T[\mathcal{C}\setminus \{c\}]=T'[\mathcal{C}\setminus \{c\}]$;
(2)  $N^+_T(c)\subseteq N^+_{T'}(c)$; and
(3) $\ts(T')\not\subseteq \ts(T)$ and $c\not\in\ts(T')$.
We construct an election as follows. Let $b$ be any arbitrary candidate in $\ts(T')\setminus \ts(T)$.
The election consists of the following votes: 1 vote defined as $T$; 1 vote defined as a tournament where
$b$ is the source; 1 vote defined as a tournament where $c$ is the source. It is clear that both $b$ and $c$ are winners.
However, by replacing the vote defined as $T$ by $T'$, $b$ gets one more approval from $T'$,
implying $c$ is no longer a winner. Therefore, in this case $\ts$-Approval is not monotonic.
\end{proof}
}

Due to Theorem~\ref{thm:concorcetconsistentmonotonicity}, to check whether $\ts$-Approval is monotonic for each $\ts\in \{\text{TC,UC,CO}\}$, we need only to investigate if $\ts$ satisfies the {\tourproperty{exclusive monotonicity}} and {\tourproperty{ENM}} criteria. Though that the {\tourproperty{monotonicity}} of $\ts$ for each $\ts\in \{\text{CO,UC,TC}\}$ is apparent and has been studied in the literature~\cite{handbookofcomsoc2015Cha3Brandt}, whether $\ts$ satisfies the two variants of the {\tourproperty{monotonicity}} criterion is not equally easy to see. In fact, we prove that among the three tournament solutions, only the top cycle satisfies the both the {\tourproperty{exclusive monotonicity}} and the {\tourproperty{ENM}} criteria. It then follows from this fact and Theorem~\ref{thm:concorcetconsistentmonotonicity} that TC-Approval is monotonic, but CO-Approval and UC-Approval are not. Our results concerning the above discussion are summarized in Lemma~\ref{lem:tocexclusivemonotonicity} and Theorem~\ref{thm:tcmonotonicucconot} shown below.

\begin{lemma}
\label{lem:tocexclusivemonotonicity}
{\Topcname} satisfies the {\tourproperty{exclusive monotonicity}}
and {\tourproperty{ENM}} criteria.
\end{lemma}

\begin{proof}
We first show that the {\topcname} is {\tourproperty{exclusive monotonic}}.
Let $T=(\mathcal{C},\succ)$ and $T'=(\mathcal{C},\succ')$ be two tournaments, and $c\in TC(T)$ be a candidate such that
(1) $T[\mathcal{C}\setminus \{c\}]=T'[\mathcal{C}\setminus \{c\}]$ and
(2) $N^+_T(c)\subseteq N^+_{T'}(c)$.
Since the {\topcname} is {\tourproperty{monotonic}}~\cite{handbookofcomsoc2015Cha3Brandt},
it holds that $c\in \topc{T'}$.
It remains to show that $\topc{T'}\subseteq \topc{T}$. 
Due to the definition of the {\topcname}, it holds that $\topc{T}\trelation \mathcal{C}\setminus \topc{T}$. Then, it holds that $\topc{T}\trelation' \mathcal{C}\setminus \topc{T}$. It directly follows that $\topc{T'}\subseteq \topc{T}$.

Now we prove that the {\topcname} satisfies the {\tourproperty{ENM}} criterion.
Let $T=(\mathcal{C},\trelation)$ and $T'=(\mathcal{C},\trelation')$ be two tournaments, and $c\in\mathcal{C}$ be a candidate such that $c\not\in TC(T)$, $T[\mathcal{C}\setminus \{c\}]=T'[\mathcal{C}\setminus \{c\}]$, and $N^+_T(c)\subseteq N^+_{T'}(c)$.
Let $CC_1,CC_2,...,CC_t$ be the maximal
strongly connected components of $T$.
It is known that for every two distinct $CC_i$ and $CC_j$ where
$\{i,j\}\subseteq \{1,2,...,t\}$, it holds that either
$V(CC_i)\trelation V(CC_j)$ or $V(CC_j)\trelation V(CC_i)$,
where $V(CC_i)$ denotes the vertices of $CC_i$. 
Moreover, there is a unique ordering $(CC_{\rho(1)},CC_{\rho(2)},...,CC_{\rho(t)})$ where
$\{\rho(1),\rho(2),...,\rho(t)\}=\{1,2,...,t\}$ such that $CC_{\rho(i)}\trelation CC_{\rho(j)}$
for every $1\leq i< j\leq t$~\cite{Moon2013}. 
Furthermore, $\topc{T}=CC_{\rho(1)}$.
Without loss of generality, assume that $c\in CC_{\rho(i)}$ for some $1< i\leq t$.
We shall show that either $\topc{T'}=\topc{T}$, or $\topc{T'}\not\subseteq \topc{T}$ and $c\in \topc{T'}$.
Due to the above discussion, if $CC_{\rho(1)}\trelation' \{c\}$ in $T'$, then $\topc{T'}=\topc{T}=CC_{\rho(1)}$. Assume now that there exists some candidate $b$ in $CC_{\rho(1)}$ such that $c\trelation' b$. We distinguish between the following cases to proceed the proof.

Case 1. $CC_{\rho(i)}=\{c\}$, or $|CC_{\rho(i)}|>1$ and $c$ is the source of $T'[CC_{\rho(i)}]$.

Let $j$ be the minimum integer such that $i>j\geq 0$ and $\{c\}\trelation' \bigcup_{i> j'>j}CC_{\rho(j')}$.
If $j=0$, then $\{c\}$ is the top cycle of $T'$. Otherwise, $(\bigcup_{1\leq j'\leq j}CC_{\rho(j')})\cup\{c\}$ is the top cycle of $T'$.

Case 2. $|CC_{\rho(i)}|>1$ and $c$ is not the source of $T'[CC_{\rho(i)}]$.

In this case, $\bigcup_{1\leq j\leq i}CC_{\rho(j)}$ is the top cycle of $T'$.

In summary, we can conclude that either $\topc{T'}=\topc{T}$, or $\topc{T'}\not\subseteq \topc{T}$ and $c\in \topc{T'}$. Thus, the {\topcname} satisfies the {\tourproperty{ENM}} criterion.
\end{proof}


Both the {\copname} and the {\ucname} do not satisfy {\tourproperty{ENM}}.
Counter-examples can be found in {\myfig{fig:counterexamplecopucENM}}.
Due to Theorem~\ref{thm:concorcetconsistentmonotonicity}, both CO-Approval and UC-Approval are not monotonic.
Due to Theorem~\ref{thm:concorcetconsistentmonotonicity}, Lemma~\ref{lem:tocexclusivemonotonicity}, and the above discussion, we have the following theorem.

\begin{figure}
\begin{center}
\includegraphics[width=0.4\textwidth]{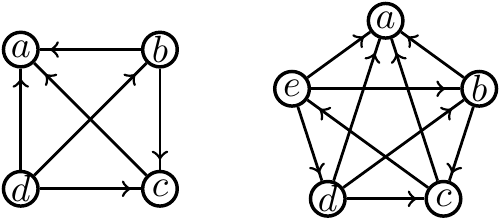}
\end{center}
\myvspace{-18pt}
\caption{The left tournament illustrates that the Copeland set does not
satisfy the {\tourproperty{ENM}} criterion. In this tournament $\{d\}$ is the Copeland set.
However, reversing the arc between $a$ and $d$ makes $\{b,d\}$ the Copeland set.
The right tournament illustrates that the uncovered set does not
satisfy the {\tourproperty{ENM}} criterion. In this tournament $\{e,c,d\}$ is the uncovered set.
However, reversing the arc between $a$ and $d$ makes $\{b,c,d,e\}$ the uncovered set.}
\label{fig:counterexamplecopucENM}
\myvspace{-17pt}
\end{figure}


\begin{theorem}\label{thm:tcmonotonicucconot}
TC-Approval is monotonic, and UC-Approval and CO-Approval are not monotonic.
\end{theorem}

Finally, we study the {Pareto optimal} criterion.
\myvspace{-5pt}

\begin{theorem}\label{thm_TC_Approval_Pareto_CO_UC_Approval_not}
TC-Approval is Pareto optimal, and CO-Approval and UC-Approval are not Pareto optimal.
\end{theorem}

\begin{proof}
We first show that TC-Approval is {Pareto optimal}.
Let $\vc=$TC-Approval. Let $\mathcal{E}=(\mathcal{C},\mathcal{T})$
be an election and $a,b\in\mathcal{C}$ be two candidates such that $a\trelation b$ in every $T=(\mathcal{C},\trelation)\in \mathcal{T}$ and $a\not\in \vc(\mathcal{E})$.
We need to prove that $b\not\in \vc(\mathcal{E})$. It is clear that if $a\not\in \topc{T}$ for some $T\in \mathcal{T}$, then $b\not\in\topc{T}$.
Therefore, $\score{a}{\mathcal{E}}{TC}\geq \score{b}{\mathcal{E}}{TC}$.
Since $a\not\in \vc(\mathcal{E})$, there is a candidate $d\in \mathcal{C}\setminus \{a\}$ such that
$\score{d}{\mathcal{E}}{TC}>\score{a}{\mathcal{E}}{TC}$. As a result,
$\score{d}{\mathcal{E}}{TC}>\score{b}{\mathcal{E}}{TC}$, implying
that $b\not\in \vc(\mathcal{E})$.
To prove that CO-Approval and UC-Approval are not Pareto optimal,
one only needs to check the election with four candidates $a,b,c,d$
and one vote with preference
$a\trelation b, \{b\}\trelation \{d,c\}, \{c\}\trelation \{a,d\}$
and $d\trelation a$. It is easy to see that $\{b, c\}$ is the winning set
in CO-Approval, but all votes (in this case only one vote) prefer $a$ to $b$.
On the other hand, $\{a,b,c\}$ is the winning set in UC-Approval, but all votes prefer $d$ to $a$.
\end{proof}

\section{Complexity of Strategic Behavior}
\begin{table*}

\begin{center}\scalebox{1.2}{
\begin{tabular}{|l|c|c|c|c|c|c|} \hline
             &  Plurality & Approval  &        TC-Approval           &         CO-Approval           &         UC-Approval        &    Evidence     \\ \hline\hline

%

CCAV         &    P       &   {\nph}    & {\myresultemp{{\nph}}}    &   {\myresultemp{{\nph}}}    & {\myresultemp{{\nph}}}  &   Theorem~\ref{thm:CCAVCCDV}  \\ \hline

CCDV         &    P       &   {\nph}    & {\myresultemp{{\nph}}}    &   {\myresultemp{{\nph}}}    & {\myresultemp{{\nph}}}  &   Theorem~\ref{thm:CCAVCCDV}  \\ \hline

CCAC         &   {\nph}     &     I     & {\myresultemp{{\nph}}}    &   {\myresultemp{{\nph}}}    & {\myresultemp{{\nph}}}  &   Theorem~\ref{thm:DCACDCDC}  \\ \hline

CCDC         &   {\nph}     &     P     & {\myresultemp{{\nph}}}    &   {\myresultemp{{\nph}}}    & {\myresultemp{{\nph}}}  &   Theorem~\ref{thm:DCACDCDC}  \\ \hline\hline

DCAV         &    P       &     P     &  {\myresultemp{\poly}}      &    {\myresultemp{\poly}}      &  {\myresultemp{\poly}}    &   Theorem~\ref{thm:polymanipulationDCAVDCDV}  \\ \hline

DCDV         &    P       &     P     &  {\myresultemp{\poly}}      &    {\myresultemp{\poly}}      &  {\myresultemp{\poly}}    &   Theorem~\ref{thm:polymanipulationDCAVDCDV}  \\ \hline

DCAC         &   {\nph}     &     P     & {\myresultemp{{\nph}}}    &  {\myresultemp{{\nph}}}     & {\myresultemp{{\nph}}}  &   Theorem~\ref{thm:DCACDCDC}  \\ \hline

DCDC         &   {\nph}     &     I     & {\myresultemp{{\nph}}}    &  {\myresultemp{{\nph}}}     & {\myresultemp{{\nph}}}  &   Theorem~\ref{thm:DCACDCDC}  \\ \hline\hline

CBRA         &    -       &     -     &  {\myresultemp{\poly}}      &   {\myresultemp{{\nph}}}    & {\myresultemp{{\wbh}}}  &   Theorem~\ref{thm:CBRA}  \\ \hline

DBRA         &    -       &     -     &  {\myresultemp{\poly}}      &   {\myresultemp{\poly}}    & {\myresultemp{{\wbh}}}  &   Theorem~\ref{thm:DBRA}  \\ \hline
\end{tabular}}
\end{center}
\caption{
A summary of the complexity of strategic voting problems.
Our results are boldfaced. The results for Plurality and
Approval are from~\protect\cite{DBLP:journals/ai/HemaspaandraHR07}.
In the table, ``{\poly}'' stands for ``polynomial-time solvable'', and ``I''
stands for ``immune''
\protect\footnotemark.
The {\wbhns} results are with respect to the number of arcs
that can be reversed in total.
The {\wbhns} results are based on {\fpt}-reductions, but not polynomial-time reductions.
It remains open whether CBRA and DBRA for UC-Approval are {\nph}.
CBRA and DBRA are not defined for Plurality and Approval. 
All results in the table apply to both the unique-winner model and the nonunique-winner model.}
\label{tab:complexitysummary}
\end{table*}
{\footnotetext{A voting system is {\it{immune}} to a constructive (resp. destructive) strategic voting problem if it is impossible to change a nonwinning (resp. winning) candidate to a winning (resp. nonwinning) candidate by performing the operations imposed in the definition of the problem.}
}

An obstacle to the fairness of voting systems is strategic behavior, potentially carried out by strategic individuals. 
For instance, a strategic individual wants to change the election
result by adding/deleting some voters/candidates, or by bribing some voters.
We refer to~\cite{DBLP:conf/birthday/BetzlerBCN12,handbookofcomsocBrandt2016,DBLP:conf/sofsem/ChevaleyreELM07}
for comprehensive surveys on this topic. In this section, we study the
complexity of strategic behavior in $\ts$-Approval.
\onlyfull{Complexity has been widely recognized as a barrier against strategic behavior~\cite{baumeisterapproval09,DBLP:journals/jair/FaliszewskiHHR09,DBLP:conf/atal/FaliszewskiHS10,DBLP:journals/ai/HemaspaandraHR07,Bartholdi92howhard}.
The point is that if it is {\nph} for the strategic individual to find out
how to successfully change the result, he might give up attacking the election.
Due to this, complexity of strategic behavior of voting systems can be also considered as a property to evaluate the voting systems~\cite{Bartholdi92howhard}.
In this paper, w}\onlyaaai{W}e assume the familiarity of complexity theory.
For readers who are not familiar with complexity theory, we refer to~\cite{garey,DBLP:journals/interfaces/Tovey02}.

We particularly study the control and the bribery problems.
In each problem, there is a {\it{strategic individual}}
who has an incentive to influence the election result by modifying the {\it{registered election}} (see later for explanation of registered election).
Depending on the situations, the strategic individual may have the {\bf{goal}} to
make a given distinguished candidate $p$ win the registered election,
or have the {\bf{goal}} to make $p$ not win the registered election.
The former case is indicated by the word {``{{\it{constructive}}}''}, and
the latter case by the word  {``{\it{destructive}}''}.
Following the convention in the literature~\cite{DBLP:journals/jair/FaliszewskiHHR09,DBLP:journals/ai/HemaspaandraHR07,Yangaamas14a},
for each problem studied in this paper,
we distinguish between the {\it{unique-winner model}} and the {\it{nonunique-winner model}}.
In the unique-winner model, winning an election means to be the unique winner, while
in the nonunique-winner model, winning an election means to be the unique winner
or to be a co-winner. 

Now we explain what modification operations the strategic individual
may perform. 


{\bf{Control.}} In the control problems studied in this paper, the strategic individual may perform 
one of the following four modification operations:  
adding/deleting at most $k$ votes/candidates, where $k>0$ is a given integer.
Therefore, the combination of the two goals and the four modification operations gives us
in total eight control problems denoted by CCAV, CCDV, CCAC, CCDC, DCAV, DCDV, DCAC and DCDC.
The first two characters ``CC''/``DC'' in the notations stand for ``constructive control''/``destructive control'',
and the last two characters ``AV''/``DV''/``AC''/``DC'' stand for
``adding votes''/``deleting votes''/``adding candidates''/``deleting candidates''. 

In the inputs of all control problems, we have a set $\mathcal{C}$ of candidates,
a list $\mathcal{T}$ of votes over $\mathcal{C}$,
a distinguished candidate $p\in \mathcal{C}$, and an integer $k>0$. A registered election consists of all registered candidates and registered votes.
In CCDV/DCDV, all candidates and votes are registered.
The question is whether the strategic individual can achieve
{\his} goal by deleting at most $k$ votes from $\mathcal{T}$.
In CCAV/DCAV, all candidates are registered, but not all votes.
In particular, a sublist $\mathcal{U}\subseteq \mathcal{T}$ of unregistered votes is given in the input.
The question is whether the strategic individual can achieve
{\his} goal by adding (registering) at most $k$ votes in $\mathcal{U}$.
In CCDC/DCDC, all candidates and votes are registered.
The question is whether the strategic individual can achieve
{\his} goal by deleting at most $k$ candidates from $\mathcal{C}\setminus \{p\}$.
It should be pointed out that the deletion of a candidate does not
affect the preference of a vote over the remaining candidates.
In CCAC/DCAC, all votes are registered but not all candidates.
In particular, a subset $\mathcal{D}\subseteq \mathcal{C}\setminus \{p\}$ of unregistered
candidates is given in the input.
The question is whether the strategic individual can achieve
his goal by adding (registering) at most $k$ candidates in $\mathcal{D}$.

The above defined control problems for many voting systems have been extensively studied in the literature. Due to the work of many researchers, the complexity of these control problems for almost all commonly used voting systems is known. A motivation of the study of control problems is that the issues of adding/deleting votes/candidates occur in many electoral settings, see, e.g.,~\cite{handbookofcomsoc2015Cha3Brandt,DBLP:journals/jair/FaliszewskiHH15} for some concrete examples. In addition, as argued in~\cite{DBLP:journals/jair/FaliszewskiHH15}, adding voters pertains to simply encouraging some agents to vote, multiplying the existing
agents, or performing false-name attacks.
We refer to~\cite{Bartholdi92howhard,handbookofcomsoc2015Cha3Brandt,DBLP:journals/jair/FaliszewskiHHR09,DBLP:journals/ai/HemaspaandraHR07,Yangaamas14b,Yangaamas14a} for further discussions on control problems.

{\bf{Bribery.}} We study two bribery problems: {\it{Constructive Bribery by Reversing Arcs}} (CBRA)
and {\it{Destructive Bribery by Reversing Arcs}} (DBRA).
In both problems we are given an election $\mathcal{E}=(\mathcal{C},\mathcal{T})$, a distinguished candidate $p\in \mathcal{C}$, and an integer $k>0$.
The question is whether the strategic individual can achieve {\his} goal
by reversing at most $k$ arcs in total in tournaments in $\mathcal{T}$.

We remark that CBRA and DBRA have already been studied under the name {\it{microbribery}}~\cite{DBLP:journals/jair/FaliszewskiHHR09}. However, the complexity of CBRA/DBRA for $\ts$-Approval has not been studied yet.

The study of bribery problems was initiated by Faliszewski, Hemaspaandra and Hemaspaandra~\shortcite{FaliszewskiHH06}, and since then many bribery problems have been proposed and studied~\cite{DBLP:conf/aaai/BredereckFNT16,DBLP:journals/jair/BredereckFNT16,DBLP:conf/aaai/DeyMN16,DBLP:journals/jair/FaliszewskiHH09,DBLP:conf/atal/KaczmarczykF16,DBLP:conf/aaai/PiniRV13,ECAI2016YangSGDistanceRestrictedBribery,AAMAS15YangSGpartybribery}.
A major motivation of such studies is that bribery behavior in voting happens in many real-world situations, such as in political elections. 
Our results concerning the complexity of control and
bribery problems are summarized in Table~\ref{tab:complexitysummary}. We achieve polynomial-time solvability results, {\nphns} results, as well as {\wbhns} results for the control and bribery problems for ${\ts}$-Approval for different tournament solutions ${\ts}$.
We compare our complexity results for ${\ts}$-Approval with the previous known results for the two most relevant voting systems Plurality and Approval.
Our results reveal that $\ts$-Approval resists more types of strategic behavior than both Plurality and Approval.
 We need the following two problems to establish our hardness results.

\medskip
\noindent{Exact 3 Set Cover} (X3C)\\
{\it{Input:}} A universal set ${U}=\{c_1,c_2,...,c_{3{\mathxxxcssize}}\}$
and a collection ${S}$ of 3-subsets of $U$.\\
{\it{Question:}} Is there an ${S}'\subseteq {S}$
such that $|{S}'|={\mathxxxcssize}$ and each $c_i\in U$
appears in exactly one set of ${S}'$?
\medskip

We assume that each element $c_i\in U$ occurs in
exactly three different 3-subsets of ${S}$. Thus,
we have that $|{S}|=3{\mathxxxcssize}$. This assumption does not change the
{\nphns} of the problem~\cite{DBLP:journals/tcs/Gonzalez85}.

A {\it{dominating set}} of a tournament $T=(V,\trelation)$ is a vertex subset $D$ of the tournament
such that for every vertex $v$ not in $D$, there is a vertex $u$ in $D$ such that
$u\trelation v$.

\medskip
\noindent{\sc{Tournament Dominating Set} {{(TDS)}}}\\
{\it{Input:}} A tournament $T=(V,\trelation)$ and an integer $k>0$.\\
{\it{Parameter:}} $k$.\\
{\it{Question:}} Does $T$ have a dominating set of size at most $k$?
\medskip

It is known that TDS is {\wbh}~\cite{DBLP:series/txcs/DowneyF13}.

\myvspace{-3pt}
\subsection{Complexity of Election Control}
In this section, we study the complexity of control by adding/deleting votes/candidates for $\ts$-Approval for difference tournament solutions $\ts$. 
We first study CCAV and CCDV. We show that both problems for $\ts$-Approval for $\ts$ being several natural tournament solutions are {\nph}, as summarized in the following theorem. Recall that both problems are polynomial-time solvable for Plurality but {\nph} for Approval~\cite{DBLP:journals/ai/HemaspaandraHR07}.

\begin{theorem}\label{thm:CCAVCCDV}
CCAV and CCDV are {\nph}  for $\ts$-Approval for every $\ts\in \{TC, UC, CO\}$,
for both the unique-winner model and the nonunique-winner model.
\end{theorem}

\begin{proof}[Sketch]
We first consider  the unique-winner model of CCAV.
The following reduction applies to every $\ts\in \{\text{TC, UC, CO}\}$.
Let $I=(U=\{c_1,c_2,...,c_{3{\kappa}}\},S=\{s_1,s_2,...,s_{3\kappa}\})$ be
an instance of the X3C problem. We create an instance {$(\mathcal{C},\mathcal{T},p\in \mathcal{C},\mathcal{U}\subseteq \mathcal{T},k)$} for the CCAV problem as follows.

{\bf{Candidates $\mathcal{C}$.}} We create in total $3\kappa+2$ candidates.
In particular, for each $c_i\in U$, we create a candidate $a(c_i)$.
In addition, we have a distinguished candidate $p$ and a dummy candidate $q$.

{\bf{Registered votes $\mathcal{T}\setminus \mathcal{U}$.}} For each $c_i\in U$, we create $\kappa-1$ votes,
each represented by a tournament where $a(c_i)$ is the source. The arcs between candidates in $\mathcal{C}\setminus \{a(c_i)\}$
are set arbitrarily. In addition, we create 1 vote defined as a tournament where $p$ is the source. The arcs between candidates in
$\mathcal{C}\setminus \{p\}$ can be set arbitrarily.

{\bf{Unregistered votes ${\mathcal{U}}$.}} For each $s_i=\{c_x,c_y,c_z\}\in S$ where $\{x,y,z\}\subseteq \{1,2,...,3\kappa\}$, we create a vote  represented by a tournament $T_{s_i}$ such that
$a(c_x),a(c_y),a(c_z),p,q$ induce a regular subtournament
and, moreover there is an arc from every candidate in the set $\{a(c_x),a(c_y),a(c_z),p,q\}$ to every candidate not in the set.
The arcs between candidates in $\mathcal{C}\setminus \{a(c_x),a(c_y),a(c_z),p,q\}$ are set arbitrarily.
Observe that for every $\ts\in \{\text{TC, UC, CO}\}$,
it holds that $\ts(T_{s_i})=\{c_x,c_y,c_z,p,q\}$. 

Finally, we set $k=\kappa$. It is easy to see that with the registered votes, each candidate
$a(c_i)$ where $i\in \{1,...,3\kappa\}$ has $\ts$-Approval score $\kappa-1$,
the distinguished candidate $p$ has $\ts$-Approval score $1$, and the
dummy candidate $q$ has $\ts$-Approval score $0$.
Observe that $q$ cannot have an equal or higher $\ts$-Approval score than
that of $p$ no matter which unregistered votes are added, since all
unregistered votes approve both $p$ and $q$.
Observe further that adding any one unregistered vote increases the $\ts$-Approval score
of some candidate $a(c_i)$ where $i\in \{1,...,3\kappa\}$ to $\kappa$.
Hence, in order to make $p$ the unique winner, we need add exactly $k$
unregistered votes. As a result, $p$ has $\ts$-Approval score $\kappa+1$ in the final election.
Moreover, for every $a(c_i)$ where $i\in \{1,...,3\kappa\}$, we can add only one unregistered vote that
approves $a(c_i)$. This happens if and only if there is an exact 3-set cover.

The {\nphns} reduction for the nonunique-winner model is similar to the
above reduction with the difference that 
we create one more registered vote for each $c_i\in U$. 

Now we consider the CCDV problem for the unique-winner model.
We construct an instance $(\mathcal{C},\mathcal{T},p\in \mathcal{C},k)$ as follows.

{\bf{Candidates $\mathcal{C}$.}} We create in total $3\kappa+1$ candidates in $\mathcal{C}$.
In particular, for each $c_i\in U$, we create a candidate $a(c_i)$.
In addition, we have a distinguished candidate $p$.

{\bf{Votes $\mathcal{T}$.}} For each $s_i=\{c_x,c_y,x_z\}\in S$ where $\{x,y,z\}\subseteq \{1,2,...,3\kappa\}$,
we create a vote represented by a tournament $T_{s_i}$ where $a(c_x),a(c_y),a(c_z)$
form a directed triangle and, moreover, there is an arc from every candidate
in $\{a(c_x),a(c_y),a(c_z)\}$ to every candidate in $\mathcal{C}\setminus \{a(c_x),a(c_y),a(c_z)\}$.
It is easy to check that $\ts(T_{s_i})=\{a(c_x),a(c_y),a(c_z)\}$ for
every $\ts\in \{\text{TC, UC, CO}\}$.
In addition, we create $3$ votes,
each represented by a tournament where $p$ is the source.

Since each $c_i$ occurs in exactly three sets in $S$,
each $a(c_i)$ where $i\in \{1,2,...,3\kappa\}$ as well as $p$ has $\ts$-Approval score $3$.
In order to make $p$ the unique winner, for each $a(c_i)$ where
$i\in \{1,...,3\kappa\}$, we need to delete one vote that approves $a(c_i)$.
This happens if and only if there is an exact 3-set cover.

The {\nphns} reduction for the CCDV problem for the nonunique-winner model is
similar to the above reduction,
with only the difference that we create one less vote corresponding to $p$,
so that the $\ts$-Approval score of $p$ in the original election is 2.
\end{proof}
}

Now we consider DCAV and DCDV.
It is known that both problems are polynomial-time solvable for Plurality and Approval~\cite{DBLP:journals/ai/HemaspaandraHR07}.
We prove that both problems are polynomial-time solvable for every $\ts$-Approval where $\ts$ is a polynomial-time computable tournament solution.
A tournament solution $\ts$ is {\it{polynomial-time computable}} if
for every tournament $T$, the set $\ts(T)$ can be calculated in
polynomial time in the size of $T$. It is well known that the
top cycle, the uncovered set and the Copeland set are
all polynomial-time computable~\cite{handbookofcomsoc2015Cha3Brandt}.
It is worth mentioning that there exist numerous tournament solutions such as the tournament equilibrium set which are not polynomial-time computable.
See~\cite{handbookofcomsoc2015Cha3Brandt} for further discussions.

\myvspace{-4pt}
\begin{theorem}\label{thm:polymanipulationDCAVDCDV}
DCAV and DCDV are polynomial-time solvable for $\ts$-Approval
such that $\ts$ is polynomial-time computable,
for both the unique-winner model and the nonunique-winner model.
\end{theorem}
\begin{proof}
To prove the theorem, we reduce DCAV (resp. DCDV) for $\ts$-Approval
where $\ts$ is a polynomial-time computable tournament solution
to the same problem for Approval in polynomial time.
In particular, given an instance of
 DCAV (resp. DCDV)  for $\ts$-Approval,
we calculate $\ts(T)$ for every vote $T$ in the instance
(for both registered and unregistered votes if applicable). Since $\ts$
is polynomial-time computable, this can be done in polynomial time.
Then, we can get an instance of DCAV (resp. DCDV)  for
Approval by taking the same candidate set, and changing
each vote originally defined as a tournament $T$ to a vote
defined as the dichotomous preference $(\ts(T),\mathcal{C}\setminus \ts(T))$,
where $\mathcal{C}$ is the candidate set.
The theorem then follows from the fact that  DCAV (resp. DCDV) 
for Approval is polynomial-time solvable, for both the unique-winner
model and the nonunique-winner model~\cite{DBLP:journals/ai/HemaspaandraHR07}.
\end{proof}

Theorem~\ref{thm:polymanipulationDCAVDCDV} implies that
DCAV and DCDV are polynomial-time solvable for
TC-Approval, CO-Approval and UC-Approval.

Now we turn our attention to control by adding/deleting candidates.
In Plurality, each vote is defined as a linear order over the
candidates and the top ordered candidate is approved.
A linear order can be considered as a transitive tournament,
where there is an arc from a candidate $a$ to another candidate $b$
if $a$ is ordered before $b$. Thus, the top ordered candidate
in the linear order is the source of the transitive tournament.
It is clear that the top cycle,
Copeland set and uncovered set of a transitive tournament consist of exactly the source
of the tournament. Hence, CCAC/CCDC/DCAC/DCDC for Plurality is
a special case of the same problem for $\ts$-Approval for every
$\ts\in \{\text{TC, UC, CO}\}$.
Since CCAC/CCDC/DCAC/DCDC for Plurality is {\nph}~\cite{DBLP:journals/ai/HemaspaandraHR07},
for both the unique-winner model
and the nonunique-winner model, the same problem for $\ts$-Approval for each
$\ts\in \{\text{TC, UC, CO}\}$ is {\nph} as well\footnote{From a parameterized complexity point of view, Yang and Guo~\shortcite{Yangaamas14b} proved that the CCDC problem for Plurality is {\wah} with respect to the number of deleted candidates even in 3-peaked elections.
}, as summarized in the following theorem.

\myvspace{-4pt}
\begin{theorem}\label{thm:DCACDCDC}
CCAC/CCDC/DCAC/DCDC is {\nph} for $\ts$-Approval for every
$\ts\in \{TC, UC, CO\}$, for both the unique-winner model
and the nonunique-winner model.
\end{theorem}

\subsection{Complexity of Bribery}
Now we study CBRA and DBRA for $\ts$-Approval for different tournament solutions $\ts$.  
Since reversing an arc may make a transitive tournament intransitive, it does not make sense to study CBRA and DBRA for Plurality and Approval.

\begin{theorem}
\label{thm:CBRA}
CBRA is {\wbh} for UC-Approval, {\nph} for CO-Approval, and polynomial-time solvable for TC-Approval,
for both the unique-winner model and the nonunique-winner model.
\end{theorem}

\begin{proof}
Yang and Guo~\shortcite{DBLP:conf/aldt/YangG13} studied a problem where
the given are a tournament and a distinguished vertex $p$ in the tournament,
and the question is whether $p$ can be made a king by reversing at most $k$
arcs. In particular, they proved that this problem is W[2]-hard with respect to $k$, by a reduction from the TDS  problem (See~\cite{DBLP:conf/aldt/YangG13} for further details). This problem can be considered as a special case of the
nonunique-winner model of the CBRA problem for UC-Approval, where the instances consist of only one vote.
To prove the {\wbhns} of CBRA for
UC-Approval for the unique-winner model, we need only to create one more vote defined as a tournament
where the distinguished candidate $p$ is the source (the arcs between other
candidates can be set arbitrarily) in the reduction in~\cite{DBLP:conf/aldt/YangG13}.

Now we prove the {\nphns} of CBRA for CO-Approval by a reduction from the X3C problem.
We first study the nonunique-winner model.
Let $(U=\{c_1,c_2,...,c_{3{\kappa}}\},S=\{s_1,s_2,...,s_{3\kappa}\})$ be
an instance of the X3C problem. We create an instance
{$(\mathcal{C},\mathcal{T},p\in \mathcal{C},k=\kappa)$}
of CBRA as follows. Without loss of generality,
assume that $k\geq 4$.

{\bf{Candidates $\mathcal{C}$.}} We create in total $6\kappa+1$ candidates in $\mathcal{C}$. 
For each $c_x\in U$, we create a candidate $a(c_x)$. For each $s_i\in S$, we create a candidate $a(s_i)$. In addition, we have a distinguished candidate $p$.

{\bf{Votes $\mathcal{T}$.}} We create in total $3k^2+4k+2$ votes. For ease of exposition, we divide the votes into three sublists $A,B,C$. The sublist $A$ consists of the following votes. For each $s_i=\{c_x,c_y,c_z\}\in S$, we create a vote defined as a tournament $H_{s_i}=(\mathcal{C},\trelation_{s_i})$ such that
(1) $a(c_x) \trelation_{s_i} a(c_y)$, $a(c_y)\trelation_{s_i} a(c_z)$, $a(c_z)\trelation_{s_i} a(c_x)$;
(2) there is an arc from $a(s_i)$ to every of $\{a(c_x),a(c_y),a(c_z)\}$;
(3) there is an arc from every $\{a(c_x),a(c_y),a(c_z)\}$ to every
candidate in $\mathcal{C}\setminus \{a(s_i),a(c_x),a(c_y),a(c_z)\}$;
(4) there is an arc from $a(s_i)$ to every candidate in
$\mathcal{C}\setminus \{a(s_i),a(c_x),a(c_y),a(c_z),a(c_{u}),a(c_v)\}$,
and an arc from each of $\{a(c_u), a(c_v)\}$ to $a(s_i)$, where $a(c_{u})$
and $a(c_v)$ are any two arbitrary candidates in
$\mathcal{C}\setminus \{a(s_i),a(c_x),a(c_y),a(c_z),p\}$; and
(5) $H_{s_i}[\mathcal{C}\setminus \{a(s_i),a(c_x),a(c_y),a(c_z)\}]$ is regular.
It is easy to verify that in $H_{s_i}$, the Copeland score of
each $a(s_i),a(c_x),a(c_y),a(c_z)$ is $6k-2$, and of each other candidate
is at most $\lceil\frac{6k-3}{2}\rceil+2$. Thus, the Copeland
set of $H_{s_i}$ is $\{a(s_i),a(c_x),a(c_y),a(c_z)\}$. Moreover, due to the large
score gap between candidates in the Copeland set and candidates not in the
Copeland set (and due to $k\geq 4$), no candidate in
$\mathcal{C}\setminus \{a(s_i),a(c_x),a(c_y),a(c_z)\}$ can be included in the
Copeland set of $H_{s_i}$ by reversing at most $k$ arcs.
The sublist $B$ consists of $k+2$ votes, each of which is defined as a tournament
such that $p$ is the source, and the subtournament induced by
$\mathcal{C}\setminus \{p\}$ is regular. Finally, the sublist $C$
consists of the following $3k^2$ votes. For each $c_x\in U$, we
create $k$ votes, each of which is defined as a tournament such that $a(c_x)$
is the source and the subtournament induced by $\mathcal{C}\setminus \{a(c_x)\}$
is regular. Notice that due to the regularity of the subtournaments
induced by all candidates except the sources in all tournaments constructed
in sublists $B$ and $C$, it is impossible to change the Copeland set of
every tournament in $B\cup C$ by reversing at most $k$ arcs.
It is easy to verify that the CO-Approval score of $p$ is $k+2$,
of each $a(c_x)$ where $c_x\in U$ is $k+3$, and of each $a(s_i)$ where $s_i\in S$ is $1$.

Now we prove the correctness of the reduction.

$(\Rightarrow:)$ Suppose that $S'\subset S$ is an exact 3-set cover, i.e.
$|S'|=\kappa=k$ and for every $c_x\in U$ there is exactly one $s\in S'$ such
that $c_x\in s$.
We shall show that we can make $p$ a winner by reversing at most $k$ arcs.
In particular, for each $s_i=\{a_x,a_y,a_z\}\in S'$, we reverse the arc from $a(s_i)$ to $a(c_u)$, where $a(c_u)$
is one candidate in $\mathcal{C}\setminus \{a(s_i),a(c_x),a(c_y),a(c_z),p\}$ such that
$a(c_u)\trelation_{s_i} a(s_i)$, as defined above. After reversing this arc,
the Copeland set of $H_{s_i}$ consists of only the candidate $a(s_i)$.
As a result, the CO-Approval score of each $\{a(c_x),a(c_y),a(c_z)\}$ decreases by one.
Since $S'$ is an exact 3-set cover, due to the construction, after reversing
all these $k$ arcs, the CO-Approval score of each $a(c_x)$ where $c_x\in U$
is $k+2$ and of each
$a(s_i)$ where $s_i\in S$ is $1$. Moreover, the CO-Approval score of $p$ is $k+2$, implying that $p$ is a winner.

$(:\Leftarrow)$ Observe that due to the score gap between every $a(s_i)$ where $s_i\in S$ and $p$, none of $a(s_i)$ has a chance to have a CO-Approval score that is equal to or higher than that of $p$ by reversing at most $k$ arcs. Moreover, as discussed above, we cannot include $p$ in the Copeland set of each tournament in $A\cup C$ by reversing at most $k$ arcs. Therefore, in order to make $p$ a winner, for each $a(c_x)$ where $c_x\in U$, we need to decrease the number of tournaments whose Copeland sets include $a(c_x)$ by at least one. Due to the above discussion, we cannot change the Copeland set of each tournament in $B\cup C$ by reversing at most $k$ arcs. Thus, the optimal solution is to reverse  arcs in tournaments in $A$. Moreover, if we attempt to reverse arcs in some tournament $H_{s_i}$ in $A$, the optimal choice is to only reverse an arc from $a(s_i)$ to a candidate $a(c_u)\in \mathcal{C}\setminus \{a(s_i),a(c_x),a(c_y),a(c_z),p\}$ such that $a(s_i)\trelation_{s_i} a(u)$ in $H_{s_i}$ (such a candidate $a(c_u)$ exists due to the construction of the votes), so that the CO-Approval score of each $a(c_x),a(c_y),a(c_z)$ decreases by one. Let $A'$ be the set of tournaments in $A$ where an arc is reversed. Let $S'=\{s_i\mid H_{s_i}\in A'\}$. Due to the above discussion, for every $c_x\in U$, there is at least one $H_{s_i}\in A'$ such that $c_x\in s_i$. Since we can reverse at most $k$ arcs, it holds that $|A'|\leq k$, implying that $S'$ is an exact 3-set cover.

To prove the unique-winner model of the problem,
we need only to create one more vote in $B$ in the above reduction.
\end{proof}

Now we study destructive bribery by reversing arcs.

\begin{theorem}\label{thm:DBRA}
DBRA is {\wbh} for UC-Approval, and polynomial-time solvable for TC-Approval and CO-Approval.
The results hold for both the unique-winner model and the nonunique-winner model.
\end{theorem}

\begin{proof}
We first prove the {\wbhns} of DBRA for UC-Approval for the nonunique-winner model.
We develop an {\fpt}-reduction from the TDS problem to the DBRA problem.
Let $(T=(V,\succ),k)$ be an instance of the {{TDS}} problem. Let $n=|V|$. We assume that $n\geq (k+1)(2k+4)$. This assumption does not affect the {W[2]-hardness} of the {{TDS}} problem\footnote{If $n\leq (k+1)(2k+3)$, we can add $(k+1)(2k+4)-n$ additional vertices to the tournament such that there is an arc from every vertex in the original tournament to every newly added vertex.
}.
We create an instance $\mathcal{I}=(\mathcal{C},\mathcal{T},p\in \mathcal{C},k)$ for DBRA as follows.

{\bf{Candidates $\mathcal{C}$.}} For each $v\in V$, we create a candidate $a(v)$. Let the distinguished candidate $p$ be any arbitrary candidate $a(w)$ such that $w$ is not a king in $T$ (in the W[2]-hardness reduction of the {TDS} problem in~\cite{fellowsfeasibility92}, there exist vertices $w$ which are not kings. Thus, such a candidate $a(w)$ is well defined\comments{we need to ensure that there is an isolated clique, say a triangle, in the graph $G$ in the proof of Theorem~4.1 in~\cite{fellowsfeasibility92}. In this case, each vertex in the isolated clique will not be a king.}). In addition, we create an additional candidate $q$. Thus, $\mathcal{C}=\{a(v)\mid v\in V\}\cup \{q\}$.

{\bf{Votes $\mathcal{T}$.}} The list $\mathcal{T}$ of votes consists of three sublists $\mathcal{T}_1,\mathcal{T}_2,\mathcal{T}_3$ of votes. The list $\mathcal{T}_1$ consists of only one vote $T_1=(\mathcal{C},\trelation_1)$, which is created first with a copy of $T$, i.e., $a(v)\trelation_1 a(u)$ in $T_1$ if and only if $v\trelation u$ in $T$. Then, we create an arc from every candidate in $\mathcal{C}\setminus \{q\}$ to $q$.
The list $\mathcal{T}_2$ consists of $2k+3$ votes, each of which is defined as a tournament such that $q$ is the source. The arcs between every two candidates in $\mathcal{C}\setminus \{q\}$ are set arbitrarily. The construction of the votes in $\mathcal{T}_3$ is a little involved. Let $A$ be any arbitrary $(k+1)(2k+3)$-subset of $\mathcal{C}\setminus \{p,q\}$, and $B=\mathcal{C}\setminus (A\cup \{p,q\})$. Let $(A_0,A_1,...,A_{2k+2})$ be any arbitrary partition of $A$ such that $|A_i|=k+1$ for every $i\in \{0,1,...,2k+2\}$. Let $f_i$ be any arbitrary one-to-one mapping from $A_i$ to $A_{(i+1)\mod (2k+3)}$ for every $i\in \{0,1,...,2k+2\}$. We create in total $2k+3$ votes in $\mathcal{T}_3$. In particular, for each $i\in \{0,1,...,2k+2\}$, we create a vote defined as a tournament $H_i$ such that

(1) there is an arc from $p$ to every candidate in $\mathcal{C}\setminus (A_i\cup \{p\})$;

(2) there is an arc from a candidate $a(v)\in A_{(i+1)\mod (2k+3)}$ to a candidate $a(u)\in A_i$ if and only if $f_i(a(u))=a(v)$;

(3) there is an arc from every candidate in $A_i$ to every candidate in $\mathcal{C}\setminus (A_i\cup A_{(i+1)\mod (2k+3)})$;

(4) there is an arc from every candidate in $\mathcal{C}\setminus \{q\}$ to $q$;

(5) $H_i[A_i]$ is isomorphic to $H_i[A_{(i+1)\mod (2k+3)}]$; and

(6) there is no source in $H_i[A_i]$.

The arcs that are not specified above are set arbitrarily.
It is clear that such a tournament can be constructed in polynomial time. Observe that $\{p\}\cup \{A_i\}$ is the uncovered set of $H_i$. \comments{every candidate in $B$ cannot reach $p$. since there is no source in $H_i$, for every candidate $a(v)\in A_{(i+1)\mod 2k+3}$ cannot reach at least one candidate in $H_i$.}
Moreover, by constructing the votes this way, it is impossible to make $q$ a king in each tournament $H_i$ in $\mathcal{T}_3$ by reversing at most $k$ arcs. \comments{since to the cost of making $q$ reach each candidate in $A_i$ is at least one arc reversal.}

Now we show the correctness of the reduction. It is easy to calculate that in the election, both $p$ and $q$ have the same UC-Approval score $2k+3$, and every other candidate has UC-Approval score at most $2$. Thus, $\{p,q\}$ is the UC-Approval winning set of the election.

$(\Rightarrow:)$ Let $D\subseteq V(T)$ be a dominating set of size at most $k$ of $T$. After reversing all arcs $\arc{a(v)}{q}$ where $v\in D$ in $T_1$, $q$ becomes a king in $T_1$. As a result, $q$ has UC-Approval score $2k+4$ and $p$ still has UC-Approval score $2k+3$, implying that $p$ is no longer a winner.

$(:\Leftarrow)$ Suppose that $\mathcal{I}$ is a {\yesins}.
Observe that due to the large UC-Approval score gap between $p$ and every candidate in $\mathcal{C}\setminus \{p,q\}$, none of $\mathcal{C}\setminus \{p,q\}$ has a chance to have a higher score than that of $p$ by reversing at most $k$ arcs. Therefore, $q$ is the only candidate which can prevent $p$ from being a winner.
Moreover, due to the above discussion, it is impossible to make $q$ a king in each vote in $\mathcal{T}_3$ by reversing at most $k$ arcs. Given that $q$ is the unique king in each tournament in $\mathcal{T}_2$, in order to prevent $p$ from being a winner, $q$ has to become a king in the first vote $T_1$ by reversing at most $k$ arcs. This happens only if there is a dominating set of $T$ of size at most $k$
.

To prove the unique-winner model, we need only to create one less tournament in $\mathcal{T}_2$. 

Now we develop \onlyaaai{a} polynomial-time algorithm\onlyfull{s} for DBRA for TC-Approval\onlyaaai{ for the nonunique-winner model}\onlyfull{ and CO-Approval}. 

{\bf{TC-Approval.}} Let $(\mathcal{C},\mathcal{T},p\in \mathcal{C}, k)$ be an instance where $p$ is a TC-Approval winner. Let $m=|\mathcal{C}|$. We assume that $m\geq 3$ (otherwise, we can easily solve the problem). To prevent $p$ from being a winner, we need to make a candidate $q$ have a higher TC-Approval score than that of $p$ by reversing at most $k$ arcs. Based on this observation, the algorithm breaks down the given instance into $m-1$ subinstances, each of which takes a candidate $q\neq p$ together with $(\mathcal{C},\mathcal{T},p\in \mathcal{C}, k)$ as the input, and asks whether $q$ can have a higher TC-Approval score than that of $p$ by reversing at most $k$ arcs. Obviously, the original instance is a {\yesins} if and only if at least one of the subinstances is a {\yesins}. It remains to develop a polynomial-time algorithm to solve each subinstance. The following claim is useful.

Claim. Let $H=(V,\trelation)$ be a tournament and $c\in V$ a candidate not in the top cycle of $H$.
Then, we can make $c$ be included in the top cycle of $H$ by reversing only one arc.

{\it{Proof of the claim.}} Let
$CC_{\rho(1)},CC_{\rho(2)},...,CC_{\rho(t)}$ be the unique ordering of the maximal strongly connected components of $H$
such that $CC_{\rho(i)}\trelation CC_{\rho(j)}$
for every $1\leq i< j\leq t$.
The top cycle
of $H$ is exactly $CC_{\rho(1)}$. Assume that
$c\in CC_{\rho(i)}$ for some $1< i\leq t$. Then,
by reversing any one arbitrary arc between $c$ and a candidate in $CC_{\rho(1)}$,
$\bigcup_{1\leq j\leq i} V(CC_j)$ becomes the
top cycle of the tournament (see the proof of Lemma~\ref{lem:tocexclusivemonotonicity} for some additional details). This completes the proof of the above claim.

Since $m\geq 3$, it is impossible to decrease the score gap between $p$ and $q$ by 2 by reversing 1 arc.
Then, due to the above claim,
to prevent $p$ from being a winner, an optimal choice is to reverse
arcs in the tournaments whose top cycles do not include $q$ in advance.
Precisely, the algorithm finds all tournaments in the subinstance where $q$ is not in the top cycle. Let $k'$ be the number of such tournaments. Then, due to the above discussion, we can increase the TC-Approval score of $q$ by $\min\{k,k'\}$ without changing the TC-Approval score of $p$, by reversing $\min\{k,k'\}$ arcs.
As a result, if $\score{q}{\mathcal{E}}{TC}+\min\{k,k'\}\geq \score{p}{\mathcal{E}}{TC}$ where $\mathcal{E}=(\mathcal{C},\mathcal{T})$, the subinstance is a {\yesins}; otherwise, it is a {\noins}.
\end{proof}

It should be noted that Papadimitriou and Yannakakis~\cite{DBLP:journals/jcss/PapadimitriouY96} devised an $O(n^{O(\log{n})})$-time algorithm for the {{TDS}} problem, which implies that the {{TDS}} problem is probably not {\nph}, unless {\np}$\subseteq$ {\sc{Dtime}}$(n^{\log{n}})$. In fact, Downey and Fellows proved the {\wbhns} of the {TDS} problem by a reduction from the {\it{Dominating Set}} problem which is both {\nph} and {\wbh}~\cite{fellowsfeasibility92}. However, the reduction is an {\fpt}-reduction but not a polynomial-time reduction---it takes $O(2^{O(k)}\cdot poly(n))$ time where $n$ is the number vertices of the given tournament and $k$ is the solution size. Hence, our reductions in the proofs of Theorems~\ref{thm:CBRA} and \ref{thm:DBRA} do not imply that DBRA and CBRA for UC-Approval are {\nph}. Whether DBRA and CBRA for UC-Approval are {\nph} remain open.

\section{Concluding Remarks}
\myvspace{-4pt}
We have studied a class of approval-based voting correspondences for
the scenario where voters may have intransitive preferences.
Each newly introduced voting correspondence $\ts$-Approval is a natural combination of the classic Approval correspondence and a well-studied tournament solution $\ts$.
In particular, each voter with preference $T$ is assumed to approve all
candidates in $\ts(T)$ and disapprove all the remaining candidates. The
winners are the ones receiving the most approvals. This class of new
voting correspondences extends the classic Approval voting to the settings where
voters have intransitive preferences. Note taht an intransitive preference is not necessarily to be cast by a single voter, but can be drawn from the preferences of voters in a subvoting, as we illustrated in the introduction. As far as we know, such $\ts$-Approval voting correspondence has not been studied in the literature.
In this paper, we first showed that $\ts$-Approval satisfies
several axiomatic properties for $\ts\in \{\text{CO,TC,UC}\}$. As a byproduct, we proposed two new concepts of monotonicity criteria of tournament
solutions, namely, the {\tourproperty{exclusive monotonicity}} and the {\tourproperty{ENM}}, and proved that the top cycle satisfies both monotonicity criteria, while the Copeland set and the uncovered set fail to satisfy  {\tourproperty{ENM}}.
Then, we investigated the complexity of
constructive/destructive control by adding/deleting voters/candidates and
constructive/destructive bribery by reversing arcs for $\ts$-Approval for $\ts\in \{\text{CO,UC,TC}\}$.
Our results reveal that $\ts$-Approval
resists more types of strategic behavior than both Plurality and Approval.
See Table~\ref{tab:complexitysummary} for a summary of our complexity results.


There remain several open questions{ for future research}. For instance, is top cycle the minimal tournament solution that satisfies {\tourproperty{neutrality}}, {\tourproperty{exclusive monotonicity}} and {\tourproperty{ENM}}? In addition, it is interesting to explore whether CBRA and DBRA for UC-Approval are {\nph}.

\bibliographystyle{plain}

\end{document}